\documentclass[letterpaper,11pt]{article}
\usepackage[utf8]{inputenc}

\pdfoutput=1

\usepackage{amsmath, amsthm, amssymb}
\usepackage[linesnumbered,vlined]{algorithm2e}
\usepackage{mathtools}
\usepackage{color-edits} 
\usepackage[most]{tcolorbox}
\usepackage{xfrac}
\usepackage{hyperref}
\usepackage{multirow}
\usepackage{caption}
\usepackage{bm}
\usepackage{newfloat}
\usepackage{enumitem}
\usepackage{xpatch}
\usepackage{soul}
\usepackage{bbm}
\usepackage{algorithmic}
\usepackage{todonotes}

\makeatletter
\newcommand{\thickhline}{%
    \noalign {\ifnum 0=`}\fi \hrule height 1.5pt
    \futurelet \reserved@a \@xhline
}
\makeatother

\usepackage[margin=1in]{geometry}

\allowdisplaybreaks

\definecolor{mygreen}{RGB}{20,120,60}

\date{}

\title{ Max-Weight Online Stochastic Matching: \\ Improved Approximations Against the Online Benchmark
}

\author{
Mark Braverman\thanks{Research supported in part by the NSF Alan T. Waterman Award, Grant
No. 1933331, a Packard Fellowship in Science and Engineering, and the
Simons Collaboration on Algorithms and Geometry.}\\
{\em Princeton University}
\and 
Mahsa Derakhshan \\
{\em UC Berkeley}
\and 
Antonio Molina Lovett\\
{\em Princeton University}
}

\newcommand{\E}[0]{\ensuremath{\mathbb{E}}}

\newcommand{\opt}[0]{\ensuremath{\textsc{OPT}_{\text{on}}}}
\newcommand{\lpopt}[0]{\ensuremath{\textsc{OPT}_{\text{LP}}}}
\newcommand{\alg}{\mathcal{WM}}

\newcommand{\inm}[1]{\ensuremath{Z}}

\renewcommand{\comment}[1]{}

\addauthor{md}{blue}

\DeclarePairedDelimiter{\abs}{\lvert}{\rvert}

\DeclarePairedDelimiter{\p}{\lparen}{\rparen}
\DeclarePairedDelimiter{\s}{\lbrack}{\rbrack}
\DeclarePairedDelimiter{\cu}{\lbrace}{\rbrace}
\DeclarePairedDelimiter{\clop}{\lbrack}{\rparen}

\usetikzlibrary{calc,positioning}

\usepackage{newunicodechar}
\newcommand{\upgreek}[1]{{\mathrm #1}}
\newcommand{\Alpha}{\upgreek A}
\newcommand{\Beta}{\upgreek B}
\newcommand{\Epsilon}{\upgreek E}
\newcommand{\Zeta}{\upgreek Z}
\newcommand{\Eta}{\upgreek H}
\newcommand{\Iota}{\upgreek I}
\newcommand{\Kappa}{\upgreek K}
\newcommand{\Mu}{\upgreek M}
\newcommand{\Nu}{\upgreek N}
\newcommand{\omnicron}{o}
\newcommand{\Omnicron}{\upgreek O}
\newcommand{\Rho}{\upgreek P}
\newcommand{\Tau}{\upgreek T}
\newcommand{\Chi}{\upgreek X}
\newunicodechar{α}{\alpha}
\newunicodechar{Α}{\Alpha}
\newunicodechar{β}{\beta}
\newunicodechar{Β}{\Beta}
\newunicodechar{γ}{\gamma}
\newunicodechar{Γ}{\Gamma}
\newunicodechar{δ}{\delta}
\newunicodechar{Δ}{\Delta}
\newunicodechar{ε}{\varepsilon}
\newunicodechar{Ε}{\Epsilon}
\newunicodechar{ζ}{\zeta}
\newunicodechar{Ζ}{\Zeta}
\newunicodechar{η}{\eta}
\newunicodechar{Η}{\Eta}
\newunicodechar{θ}{\theta}
\newunicodechar{Θ}{\Theta}
\newunicodechar{ι}{\iota}
\newunicodechar{Ι}{\Iota}
\newunicodechar{κ}{\kappa}
\newunicodechar{Κ}{\Kappa}
\newunicodechar{λ}{\lambda}
\newunicodechar{Λ}{\Lambda}
\newunicodechar{μ}{\mu}
\newunicodechar{Μ}{\Mu}
\newunicodechar{ν}{\nu}
\newunicodechar{Ν}{\Nu}
\newunicodechar{ξ}{\xi}
\newunicodechar{Ξ}{\Xi}
\newunicodechar{ο}{\omnicron}
\newunicodechar{Ο}{\Omnicron}
\newunicodechar{π}{\pi}
\newunicodechar{Π}{\Pi}
\newunicodechar{ρ}{\rho}
\newunicodechar{Ρ}{\Rho}
\newunicodechar{σ}{\sigma}
\newunicodechar{Σ}{\Sigma}
\newunicodechar{τ}{\tau}
\newunicodechar{Τ}{\Tau}
\newunicodechar{υ}{\upsilon}
\newunicodechar{Υ}{\Upsilon}
\newunicodechar{φ}{\varphi}
\newunicodechar{Φ}{\Phi}
\newunicodechar{χ}{\chi}
\newunicodechar{Χ}{\Chi}
\newunicodechar{ψ}{\psi}
\newunicodechar{Ψ}{\Psi}
\newunicodechar{ω}{\omega}
\newunicodechar{Ω}{\Omega}

\newcommand{\sm}{\setminus}
\newcommand{\dc}{\dotsc}

\SetKwRepeat{Do}{do}{while}

\DeclareMathOperator*{\argmax}{arg\,max}
\renewcommand{\b}[1]{\ensuremath{\bm{\mathrm{#1}}}}

\renewcommand{\epsilon}[0]{\ensuremath{\varepsilon}}

\let\originalleft\left
\let\originalright\right
\renewcommand{\left}{\mathopen{}\mathclose\bgroup\originalleft}
\renewcommand{\right}{\aftergroup\egroup\originalright}

\addauthor{am}{mulberry}    
\addauthor{md}{red}    

\newtheorem{theorem}{Theorem}

\newtheorem{lemma}{Lemma}[section]

\newtheorem{observation}[lemma]{Observation}

\makeatletter
\def\thm@space@setup{%
  \thm@preskip= 0.2cm
  \thm@postskip=\thm@preskip 
}
\makeatother

\definecolor{mygreen}{RGB}{20,155,20}
\definecolor{myred}{RGB}{195,20,20}
\definecolor{linkcolor}{RGB}{0,0,230}
\definecolor{mylightgray}{RGB}{230,230,230}
\definecolor{verylightgray}{RGB}{240,240,240}
\definecolor{commentcolor}{RGB}{120,120,120}
\definecolor{mulberry}{RGB}{197, 75, 140}

\SetCommentSty{mycommfont}

\hypersetup{
     colorlinks=true,
     citecolor= mygreen,
     linkcolor= myred,
     urlcolor= mygreen
}

\newcommand{\mc}[1]{\ensuremath{\mathcal{#1}}}

\newcounter{myalgctr}

\newenvironment{tbox}{
\par\addvspace{0.2cm}
\begin{tcolorbox}[width=\textwidth,
                  enhanced,
                  boxsep=2pt,
                  left=1pt,
                  right=1pt,
                  top=4pt,
                  boxrule=1pt,
                  arc=0pt,
                  colback=white,
                  colframe=black,
                  unbreakable
                  ]
}{
\end{tcolorbox}
}

\newenvironment{tboxh}{
\par\addvspace{0.2cm}
\begin{tcolorbox}[width=\textwidth,
                  enhanced,
                  boxsep=2pt,
                  left=1pt,
                  right=1pt,
                  top=4pt,
                  boxrule=1pt,
                  arc=0pt,
                  colback=white,
                  colframe=black,
                  unbreakable,
                  float=t
                  ]
}{
\end{tcolorbox}
}

\newenvironment{graytbox}{
\par\addvspace{0.1cm}
\begin{tcolorbox}[width=\textwidth,
                  enhanced,
                  frame hidden,
                  boxsep=5pt,
                  left=1pt,
                  right=1pt,
                  top=2pt,
                  bottom=2pt,
                  boxrule=1pt,
                  arc=0pt,
                  colback=mylightgray,
                  colframe=black,
                  breakable
                  ]
}{
\end{tcolorbox}
}

\newcommand{\tboxhrule}[0]{\vspace{0.1cm} \hrule \vspace{0.2cm}}

\newenvironment{titledtbox}[1]{\begin{tbox}#1 \tboxhrule}{\end{tbox}}
\newenvironment{titledtboxh}[1]{\begin{tboxh}#1 \tboxhrule}{\end{tboxh}}

\newenvironment{tboxalg}[2][]{\refstepcounter{myalgctr}\begin{titledtbox}{\textbf{Algorithm \themyalgctr}#1\textbf{.} #2}}{\end{titledtbox}}

\begin{document}
\maketitle

\thispagestyle{empty}
In this paper, we study max-weight stochastic matchings on online bipartite graphs under both vertex and edge arrivals. We focus on designing polynomial time approximation algorithms with respect to the online benchmark, which was first considered by Papadimitriou, Pollner, Saberi, and Wajc [EC’21]. 

In the vertex arrival version of the problem, the goal is to find an approximate max-weight matching of a given bipartite graph when the vertices in one part of the graph arrive online in a fixed order with  independent chances of failure. Whenever a vertex arrives we should decide, irrevocably, whether to match it with one of its unmatched neighbors or leave it unmatched forever.  There has been a long line of work designing approximation algorithms for different variants of this problem with respect to the offline benchmark (prophet). 
Papadimitriou et al., however, propose the alternative {\em online} benchmark and show that considering this new benchmark allows them to improve the $0.5$ approximation ratio, which is the best ratio achievable with respect to the offline benchmark. They provide a $0.51$-approximation algorithm which was later improved to $0.526$ by Saberi and Wajc [ICALP'21]. The main contribution of this paper is designing a simple algorithm with a significantly improved approximation ratio of $(1-1/e)$ for this problem.

We also consider the edge arrival version in which, instead of vertices, edges of the graph arrive in an online fashion with independent chances of failure. Designing approximation algorithms for this problem has also been studied extensively with the best approximation ratio being $0.337$ with respect to the offline benchmark. This paper, however, is the first to consider the online benchmark for the edge arrival version of the problem. For this problem, we provide a simple algorithm with an approximation ratio of $0.5$ with respect to the online benchmark.

{

\hypersetup{
     linkcolor= black
}

\thispagestyle{empty}
\clearpage
}

\setcounter{page}{1}

\section{Introduction}
%
%

The extensive literature on online Bayesian selection algorithms mainly focuses on the {\em competitive ratio}. That is, how well the algorithm performs against the optimal {\em offline} solution. Competing against such a strong benchmark often leads to pessimistic outcomes. For example, it is well-known that even for the single item version of the online Bayesian selection problem, the prophet inequality problem, no online algorithm can be better than $\sfrac{1}{2}$-competitive.

Another natural objective would be to compete with the best {\em online} solution. For many variants of online Bayesian selection problems (when the input is generated stochastically) one can write a  dynamic program that makes the best decision at any point --- hence the objective function is well defined. However, these algorithms are rarely computationally efficient. Indeed, Papadimitriou, Pollner, Saberi, and Wajc~\cite{DBLP:conf/sigecom/PapadimitriouPS21} show that a variant of the {\em online stochastic matching problem} is PSPACE-hard to approximate within some small constant.  Thus, they initiate studying approximation algorithms for this problem with respect to the online benchmark. That is the solution found by an algorithm that has unlimited computational power, but is unaware of the part of the input that has not arrived.


In this paper, we study max-weight stochastic matchings on online bipartite graphs under both vertex and edge arrivals. Our main focus is on designing polynomial time approximation algorithms with respect to the online benchmark.

\paragraph{\textbf{The Vertex arrival model.}}The goal in this problem is to find a large-weight matching of a  bipartite graph when vertices in one part of the graph are online, arriving in a fixed order, each with an independent chance of failure. The vertices in the other part are present from the beginning thus referred to as the offline vertices. The graph, the arrival order of the online vertices, and their chances of failure are known from the beginning. The only unknown is whether a vertex actually arrives or if it fails. If a vertex does not arrive (i.e., fails), we do nothing about it. Otherwise, we either match it irrevocably to one of its unmatched neighbors or leave it unmatched forever. Papadimitriou et al.~\cite{DBLP:conf/sigecom/PapadimitriouPS21}  refer to this problem as the RideHail problem due to its applications in ride hailing. However, it also models scenarios in other types of matching markets such as labor markets, online advertising, etc.

There has been a long line of work designing approximation algorithms for this problem (and its variants) with respect to the offline benchmark (see~\cite{DBLP:conf/stoc/KarpVV90, DBLP:conf/sigecom/AlaeiHL12, DBLP:conf/focs/GamlathKMSW19, DBLP:journals/geb/KleinbergW19, DBLP:conf/sigecom/EzraFGT20, DBLP:conf/sigecom/PapadimitriouPS21}  and the references within.) The best known algorithm with respect to this benchmark achieves a tight $1/2$ approximation ratio~\cite{DBLP:conf/soda/FeldmanGL15}.  
In their recent work, Papadimitriou et al.~\cite{DBLP:conf/sigecom/PapadimitriouPS21} show that this ratio can be improved to $0.51$ if one considers the online benchmark instead. The online benchmark here is defined as a max-weight matching found by an algorithm that has unlimited computational power but does not know the arrival/failure of the future vertices. This approximation ratio was later improved to $0.526$ using a machinery developed by Saberi and Wajc~\cite{DBLP:conf/icalp/SaberiW21} for an online edge coloring problem. In this work, we design a simple algorithm with a significantly improved approximation ratio of $(1-1/e)\approx0.632$ with respect to the online benchmark.

\vspace{4mm}
\begin{graytbox}
\noindent \textbf{Result 1.} (See Theorem~\ref{thm:main1}) There exists a polynomial time algorithm for the online bipartite stochastic matching problem under (one-sided) vertex arrivals which finds a matching of weight at least a $(1-1/e)$ fraction of the one found by the best online algorithm.
\end{graytbox}

Similar to Papadimitriou et al.~\cite{DBLP:conf/sigecom/PapadimitriouPS21}, we also consider a more general version of the problem, which is also studied by~\cite{DBLP:conf/soda/FeldmanGL15, DBLP:conf/sigecom/EzraFGT20, DBLP:journals/siamcomp/DuttingFKL20}, where upon arrival of a vertex, weights of its edges are drawn from a known joint distribution. However, weights of edges incident to  different online vertices are still independent. In Section~\ref{section:general} we explain how our algorithm and analysis can be extended to get the same approximation ratio of $(1-1/e)$ for this more general problem.

\paragraph{\textbf{The Edge arrival model.}} The only difference between this problem and the vertex arrival version is that here, instead of the vertices, edges are online. Similarly, the goal in this problem is to find a large-weight matching of a bipartite graph when edges are online, arriving in a fixed order, each with an independent chance of failure. The graph, the arrival order of the edges and their chances of failure are known from the beginning. The only unknown is whether an edge actually arrives or if it fails. If an edge does not arrive (i.e., fails), we do nothing about it. Otherwise, we decide irrevocably whether to add it to our matching or not. See~\cite{DBLP:conf/icalp/GravinTW21} for potential applications of this problem.

Designing approximation algorithms for the edge arrival version of the problem has also been studied extensively (see \cite{DBLP:journals/algorithmica/BuchbinderST19, DBLP:conf/ec/GravinW19, DBLP:conf/sigecom/EzraFGT20,  DBLP:journals/siamcomp/FeldmanSZ21} and the references within), with the best approximation ratio being $0.337$ with respect to the offline benchmark~\cite{DBLP:conf/sigecom/EzraFGT20}. It is also known that with respect to this benchmark it is not possible to achieve an approximation ratio better than $\sfrac{4}{9}$~\cite{DBLP:conf/ec/GravinW19}. This paper, however, is the first to consider the online benchmark for the edge arrival version of the problem. For this problem, we provide a simple algorithm with an approximation ratio of $0.5$ with respect to the online benchmark. 
\vspace{4mm}
\begin{graytbox}
\noindent \textbf{Result 2.} (See Theorem~\ref{thm:main2}) There exists a polynomial time algorithm for the online bipartite stochastic matching problem under edge arrivals which finds a matching of weight at least a $1/2$ fraction of the one found by the best online algorithm.
\end{graytbox}

\subsection{Our Techniques}
For both vertex and edge arrival versions of the problem, we design LP-based algorithms consisting of an LP and a rounding procedure. The LP which we borrow from~\cite{torrico2017dynamic} outputs a fractional solution \b{x}, where for any edge $e$, $x_e$ can be interpreted as the probability of this edge joining the matching.  Papadimitriou et al.~\cite{DBLP:conf/sigecom/PapadimitriouPS21} also use the same LP and furtur give a lower-bound of $0.875$ for its integrality gap.  Other than the basic matching constraints, the LP has an additional natural constraint which is crucial for separating the online and offline solutions. This constraint relies on the fact that whether a vertex/edge fails or not is independent of any decision made by the algorithm for the vertices/edges arriving before that. Thus, for instance, in the vertex arrival version, this constraint states that for any edge $e=(v, u)$, the offline vertex $u$ should be unmatched with probability at least $x_e/p_v$ before the arrival of vertex $v$. Here, $p_v$ denotes the probability of $v$ not failing (i.e. arriving).  We explain this in more detail in Section~\ref{section:alg}. In this section, we focus on discussing our rounding procedure for the vertex arrival model to present the flavor of our work.

\paragraph{\textbf{Our rounding procedure.}} To round the solution of the LP, we design a simple random rounding procedure.  Upon arrival of a vertex $v$, it may receive matching proposals from its unmatched neighbors. If it receives any proposals, it accepts the best one (the one with the largest weight) and joins the matching. Otherwise, it remains unmatched. The process of sending proposals is as follows: When $v$ arrives, any of its unmatched neighbors decides independently at random whether to send a proposal. The probability of sending these proposals is set in a way that for any edge, the probability of it being proposed is lower-bounded by $x_e$. This is achievable in particular due to the LP constraint discussed above.

 To be able to highlight the properties of our algorithm which allow us to improve the algorithm proposed by Papadimitriou et al.~\cite{DBLP:conf/sigecom/PapadimitriouPS21}, we will give a brief overview of their algorithm below.

\paragraph{\textbf{Algorithm proposed by Papadimitriou et al.}} Let us emphasize that this is just a brief and paraphrased overview of the algorithm proposed in~\cite{DBLP:conf/sigecom/PapadimitriouPS21} which we include for the sake of comparison. They start with the same LP that we use. However, their rounding procedure is different.  Upon arrival of an online vertex $v_t$, it picks one of its neighbors randomly proportional to the probabilities given by the LP and sends a proposal to it. If the neighbor is unmatched it accepts the proposal with some probability and matches with $v_t$. These probabilities are set in a way that each edge joins the matching with probability at least $0.51x_e$. However, some vertices are not able to satisfy this for their edges by sending only a single proposal. The algorithm gives such vertices a second chance to send another proposal if their first proposal is not accepted, and this allows them  to guarantee a matching probability of $0.51x_e$ for all the edges. In the rest of the paper, we will refer to this algorithm as PPSW (the authors’ initials.)
\\

As the first difference, in PPSW the offline vertices receive the proposals and need to decide whether to accept a proposal without knowing their future proposals. In our algorithm however, since  the online vertices are the ones receiving the proposals they can make a decision while knowing all their options. This is particularly helpful since edges are weighted and an online vertex has the option of picking the one with the highest weight. With this advantage, however, comes a new challenge. We cannot guarantee that all the edges will join the matching with a large probability. Indeed, some low-weight edges may have a very small probability. To overcome this, instead of analyzing the rounding loss for any edge, we lower-bound the loss imposed on any online vertex due to the rounding procedure. 

The second difference is in the number of proposals a vertex can send. We do not limit the number of proposals a vertex can send. Indeed, a vertex can send proposals as long as it is unmatched. PPSW on the other hand, imposes the limit of two proposal per any online vertex. This is due to the way their analysis works. The main meat of their analysis is upper-bounding the positive correlation between the matching status of the offline vertices, and allowing the vertices to send many proposals worsens the correlation. Let us first explain why absence of positive correlation is desirable.  A key part of our analysis is proving that our algorithm satisfies the following property: whenever an online vertex $v$ arrives, the probability of it {\em not} receiving any proposals from any subset $S$ of its neighbors is at most \begin{align}  \label{eq:kjfrej} \prod_{u\in S} (1-x_{(v,u)}/p_v).\end{align} 
As mentioned above, the probability of $v$ receiving a proposal from any of its neighbors $u$ is at least $x_{(v,u)}$. Since $v$ arrives with probability $p_v$, the probability of it {\em not} receiving a proposal from a neighbor $u$ is at most $ (1-x_{(v,u)}/p_v)$. Thus, property~\eqref{eq:kjfrej} follows directly if we were allowed to assume that the matching status of the vertices in $S$ are independent when $v$ arrives. It is not complicated to show that the same holds if they are not positively correlated. Unfortunately though, trying to prove this key property through positive correlation fails as we show via an example (See Section~\ref{section:tight}.) that  these events can indeed be negatively correlated. However,  we are still able to prove the existence of this property via a different method without concerning ourselves with the correlation between the matching status of the offline vertices. We even show that our analysis is tight for an instance of the problem. Interestingly, in this instance our algorithm does not cause any positive correlation between the matching status of the offline vertices.
(See Section~\ref{section:correlation}.) This all means that  positive correlation by itself is not the enemy. It only hurts us if it decreases the probability of a vertex receiving at least one proposal in comparison to the case of its neighbors being independent. Our approach is to carefully lower-bound the probability of this event using simple mathematical tools. See Lemma~\ref{lemma:no-proposals-bound} for more details.

\subsection{Further Related work}

As we mentioned before, most of the literature  on online Bayesian selection focuses on designing algorithms with respect to the offline benchmark which is often referred to as the prophet. 
It would be impossible to do justice to this extensive literature in this amount of space, thus we just briefly outline some of the most relevant works here. The study of prophet inequality problem, the single item version of the online Bayesian selection problem, was  initiated by Krengel and Sucheston~\cite{krengel1978semiamarts} who give an algorithm with competitive ratio of $1/2$. Their seminal work was a starting point for studying more general versions of online Bayesian selection problems.  The ones most related to multi-item prophet inequalities under matroid constraints~\cite{DBLP:journals/geb/KleinbergW19}, stochastic matching under vertex arrivals~\cite{DBLP:conf/stoc/KarpVV90, DBLP:conf/sigecom/AlaeiHL12, DBLP:conf/focs/GamlathKMSW19, DBLP:journals/geb/KleinbergW19, DBLP:conf/sigecom/EzraFGT20, DBLP:conf/sigecom/PapadimitriouPS21} and stochastic matching under edge arrivals~\cite{DBLP:journals/algorithmica/BuchbinderST19, DBLP:conf/ec/GravinW19, DBLP:conf/sigecom/EzraFGT20,  DBLP:journals/siamcomp/FeldmanSZ21}. 

It is worth mentioning that the connection between prophet inequalities and algorithmic mechanism design first discovered by Hajiaghayi, Kleinberg and Sandholm~\cite{DBLP:conf/aaai/HajiaghayiKS07}, has played a significant role in motivating the study of approximation algorithms for these online Bayesian selection problems. For more detailed related work on the topic of prophet inequality and also its relations to algorithmic mechanism design see ~\cite{DBLP:journals/sigecom/Lucier17, DBLP:journals/sigecom/CorreaFHOV18}.

Max-weight matching on stochastic graphs has also been studied extensively under the query model. That is, similar to our model, there is an underlying stochastic graph, however, to know whether an edge exists it should be queried. Some works focus on having a small number of queries~\cite{DBLP:conf/soda/YamaguchiM18, DBLP:conf/sigecom/BehnezhadR18, DBLP:conf/stoc/BehnezhadDH20, DBLP:conf/focs/BehnezhadD20} while others require any queried and realized edge to join the matching~\cite{DBLP:conf/icalp/ChenIKMR09, DBLP:conf/icalp/CostelloTT12, DBLP:journals/algorithmica/BansalGLMNR12, DBLP:conf/soda/GamlathKS19}. The main application of these models is for environments with costly queries such as organ exchange markets.

\paragraph{\textbf{Paper organization.}} \noindent  The rest of the paper is organized as follows. In Section~\ref{section:pre}, we provide a formal definition of our problems and some notions that we will use throughout the paper. Section~\ref{sec:vtx} is about the vertex arrival model. In this section, we first  present the algorithm and its analysis in~\ref{section:alg} and ~\ref{sec:analysis} respectively. Later, in \ref{section:correlation}, we provide an example for which our algorithm causes positive correlation between the matching status of the offline vertices, and in~\ref{section:tight} we show that the analysis of our algorithm is tight. Further, in~\ref{section:general} we explain how our results can be extended to a more general version of the problem where the weights of the edges connected to any online vertex are drawn from a joint distribution. 
Finally, we discuss the edge arrival version of the problem in Section~\ref{sec:edge} with the algorithm and its analysis being in~\ref{sec:edge-algorithm} and~\ref{section:edge-analysis} respectively.

\section{Preliminaries}\label{section:pre}
We are given a bipartite graph $G=(A, B, E)$ and  a weight $w_e$ for each edge $e\in E$. In the vertex arrival model, we also have a probability $p_v$ for each $v\in A$ and a fixed order $(v_1, \dots, v_{|A|})$ over the vertices in $A$. Vertices in $B$ are initially present, but vertices in $A$ arrive online. At any time $t$, with probability $p_{v_t}$ vertex $v_t$ arrives (or is realized). If it does, we are allowed to match $v_t$ irrevocably to one of its unmatched neighbors or else commit to leaving it unmatched forever.
If it does not arrive, we do nothing at time $t$.  In the edge arrival model, similarly, we are also given a probability $p_e$ for each $e\in E$ and a fixed order $(e_1, \dots, e_{|E|})$ over the edges. At any time $t$, with probability $p_{e_t}$ edge $e_t$ arrives (or is realized). If it does, we should decide irrevocably whether to add it to our matching or lose it forever. Under both arrival models, the goal is to maximize the total weight of the edges we add to the matching. 

In this paper, our focus is on designing a polynomial time algorithms for the above problem under both arrival models. We say that an algorithm $\mc{A}$ is an $\alpha$-approximation if for any instance $\mc I$ of the problem, it satisfies
$$\E[\mc{A(I)}] \ge \alpha\E[\opt({\mathcal I})] ,$$
where $\opt$ is the optimal online algorithm. That is an algorithm that has unlimited computational power, but its knowledge about the arrival of future vertices/edges is the same as ours.

For ease of notation, for any pair of edges $e_1$, $e_2$, we say $e_1<e_2$ if $e_1$ arrives before $e_2$. We also use $e_1 < t$ to mean $e_1$ arrives before time $t$. (Note that in the vertex arrival model, an edge arrives whenever its online end-point arrives.) Also, when it is clear from context, we will use $t$ to refer to the vertex $v_t$ (or edge $e_t$), arriving at time $t$.
Finally, we write our edges as ordered pairs, meaning that for any edge $(v,u)\in E$ we always have $v\in A$ and $u\in B$.

\section{Vertex Arrivals}\label{sec:vtx}

\subsection{The Algorithm}\label{section:alg}

We begin by writing a linear program that attempts to model the optimal (omnipotent) online algorithm's behavior. This LP is also used by Papadimitriou et al., however, for the sake of self-containment we explain it in detail here.
For any edge $e\in E$, we have a variable $x_e$ which represents the probability of $e$ joining the matching in $\opt$. Here, the randomness can be over both the stochastic arrivals of the vertices and any random decisions made by the algorithm.
We claim then that such $x_e$ are feasible for the following LP:
\begin{align} 
\max_{\mathbf{x}} \qquad &\sum_{e\in E} w_ex_e	\,,& \label{eq:LP-objfun} \\
\textrm{s.t.} \qquad & \sum_{e\ni v} x_e \le p_v &  \forall\, v\in A  \,, \label{eq:LP-topbound} \\
& \sum_{e\ni u} x_e \le 1 &  \forall\, u\in B  \,, \label{eq:LP-botbound} \\
& p_v\cdot (1-\sum_{\mathclap{e'\ni u, e'<e}} x_{e'}) \ge x_{e} &  \forall\, e=(v,u)\in E  \,, \label{eq:LP-main} \\
&x_e \ge 0& \forall e \in E \,. \label{eq:LP-nonneg}
\end{align}

The first two constraints (\ref{eq:LP-botbound} and \ref{eq:LP-topbound}) are standard matching constraints since each vertex can be incident to at most one edge in the matching, and each $v\in A$ is unmatched with probability at least $1-p_v$ (when it fails). 
Constraint~\ref{eq:LP-main} is however special to the online solution. It asserts that for any edge $(v_t, u)$ the probability of $u$ being unmatched  before time $t$ (the left-hand side) is at least $x_e/p_{v_t}$.
This is due to the fact that  arrival of vertex $v_t\in A$  is independent of whether $u$ is matched before time $t$. (All vertices arrive independently, and a non-omniscient algorithm must have made all matching decisions independently of future arrivals.) If this constraint is not satisfied then $u$ is unmatched with probability less that $x_e/p_{v_t}$. In this case, the probability that $v_t$ arrives and $u$ is still unmatched by time $t$ is less than $x_e$, contradicting the definition of $x_e$.

\begin{observation}[\cite{torrico2017dynamic}]\label{obs:optlp}
	Let \b{x} be an optimal solution of the LP. We have $\lpopt \geq \opt$ where $\lpopt = \sum_{e\in E} w_ex_e$.
\end{observation}

\subsection{The Rounding Procedure}
The next step of the algorithm is rounding the fractional solution of the LP. 
For that, we design a simple rounding procedure that given any optimal solution of the LP (which is a fractional matching), outputs an integral matching. Later we will prove that the output of this algorithm is a $(1-1/e)$-approximate solution. 

Our rounding procedure is very simple and natural. Whenever a vertex $v_t\in A$ arrives, we construct a random subset $P$ of its unmatched neighbors as potential matches. Any unmatched neighbors decide independently at random whether to send a matching proposal to $v_t$ and join subset $P$. 
 If $v_t$ receives at least one proposal (i.e., if $P$ is nonempty), it accepts the best one (the one with the largest weight) and joins the matching. Otherwise, it remains unmatched forever. In our algorithm, we set the probability of sending proposals in a way that for  any edge $e=(v_t, u)\in E$, it results in  
\begin{align*} \Pr[v_t \text{ receives a proposal from } u]\geq x_e. \end{align*} 
This is achievable thanks to the constraint~\ref{eq:LP-main} of the LP which separates the online and offline benchmarks. In other words, it is not possible to satisfy this inequality for any arbitrary fractional matching, and this is where we use the fact that we are competing with the best online algorithm. To be able to satisfy this property, whenever $v_t$ arrives and $u$ is unmatched we need  $u$ to send a proposal to $v_t$ with probability at least  $$\frac{x_e}{p_t \Pr[u \text{ is unmactched}]}.$$
This is of course achievable only if this number is not larger than one, which will be shown as a consequence of our analysis.
\vspace{3mm}

\begin{tboxalg}{Rounding Procedure}\label{alg}
\begin{algorithmic}[1]
\STATE Let $\b x$ be an optimal solution of the LP.
\STATE Let $M\gets \emptyset$ be a matching of $E$.
\FOR{$t\in |A|$}
\STATE $v\gets v_t$
\STATE Let set $N_v$ denote neighbors of vertex $v$ in graph $G$. 
\STATE $P\gets \emptyset$.
\STATE For any vertex $u\in N_v$, define $\alpha_{u}=\sum_{e\ni u, e<(v,u)} x_{e}.$ \label{alg:defalpha}
\STATE For any vertex $u\in N_v$, if $u$ is matched in $M$ and $x_{(v,u)}>0$, then with probability $\frac{x_{(v,u)}}{p_v(1-\alpha_u)}$ add edge  $(v,u)$ to set $P$ independently. \label{alg:samp-line}
\IF{$P$ is non-empty and $v_t$ is realized}
\STATE Add edge $\argmax_{e\in P} w_e$ to matching $M$.
\ENDIF
\ENDFOR
\STATE Return matching $M$. 
\end{algorithmic}
\end{tboxalg}

\vspace{2mm}

\subsection{The Analysis}\label{sec:analysis}
The purpose of this section is proving that Algorithm~\ref{alg} finds a $(1-1/e)$-approximate matching. Before proceeding with our analysis,  we need to define some notations. In the rest of the paper, we use $\alg$ to represent the weighted matching outputted by Algorithm~\ref{alg}. Moreover, for any vertex $v$, if it is matched in $\alg$, (i.e., $v\in \alg$), we use $\alg(v)$ to represent the weight of its matching edge in $\alg$. Note that $\alg$ and $\alg(v)$ are both random variables. For any vertex $u\in B$ and any time $t$, we define 
$$\alpha_{t,u}:= \sum_{e\ni u,e<(v_t,u)} x_{e}.$$
Finally, for a given subset of vertices $S\subset B$, we define $E^S_t$ to be the event in which all the vertices from $S$ are matched before time $t$, and $F^S_t$ to similarly be the event that all vertices in $S$ are still free (unmatched) just before time $t$.

In our calculations, we will make use of the following lemma. However, to preserve the flow of the paper, we defer its proof to Section~\ref{apx:omitted}.

\begin{lemma} 
  \label{lemma:incl-excl}
  Let $S\subset B$ be a set of vertices.
  Suppose we associate each vertex $u\in S$ with a number $w_u\in\mathbb R$. Then
  \begin{align*}
    \sum_{X\subset S}\Pr\s*{E^{S\sm X}_t\cap F^X_t}\prod_{u\in X}w_u
    &= \sum_{X\subset S}\Pr\s*{E^{S\sm X}_t}\p*{\prod_{u\in X}w_u}\prod_{u\in S\sm X}(1-w_u)
      .
  \end{align*}
\end{lemma}

We can now begin our analysis with a crucial property of our algorithm. That is upper-bounding the probability of all the vertices in $S$ being matched before time $t$ for any $S\subset B$.
\begin{lemma}\label{lemma:corr}
At any time $t$, for any subset of vertices $S\subset B$, we have
\begin{align}\Pr[E^S_t]\leq \prod_{u\in S} \alpha_{t,u}. \label{eq:thm-corr-main}\end{align}
\end{lemma} 

\begin{proof}
  We use proof by induction on $t$.
  Our claim holds for the base case of $t=1$ as for $t=1$, both sides of Equation~\ref{eq:thm-corr-main} equal to zero for nonempty $S$ (and one for $S=\emptyset$).
  Assuming that for some $t\geq 1$, this equation holds, we will prove it for $t+1$.
  In other words, we will prove
  \begin{align}\Pr[E^S_{t+1}]\leq \prod_{u\in S} (\alpha_{t,u}+x_{(t,u)}), \label{eq:thm-corr-main-wts} \end{align}
  for any $S \subset B$.
  For now, we assume that \begin{align}\alpha_{t+1,u}=\alpha_{t,u}+x_{(v_t,u)}<1 \label{eq:ass-unsat}\end{align} holds for all $u\in S$.
  For such $u$, we may define \begin{align}\beta_{t,u}=\frac{x_{(v_t,u)}}{1-\alpha_{t,u}}\in\clop{0,1}.\label{eq:beta-bound}\end{align}
  We aim to show that $S$ satisfies Equation~\ref{eq:thm-corr-main-wts}.
  For $E^S_{t+1}$ to occur, either all of $S$ was matched already before time $t$, or some vertex $u\in S$ was matched exactly at time $t$, with the others matched before.
  This lets us compute
  \begin{align}
    \Pr[E^S_{t+1}]
    &= \Pr[E^S_t]+\sum_{u\in S}\Pr\s*{E^{S\setminus\{u\}}_t\cap F^{\{u\}}_t}\Pr\s*{(v_t,u)\in M\mid E^{S\sm\cu u}_t\cap F^{\cu u}_t} \nonumber \\
    &\le \Pr[E^S_t]+\sum_{u\in S}\Pr\s*{E^{S\setminus\{u\}}_t\cap F^{\{u\}}_t}\frac{x_{(v_t,u)}}{p_v(1-\alpha_{t,u})}p_v \label{eq:corr-line-expand} \\ \tag{to be matched, $(v_t,u)$ must have been added to $P$ and $v_t$ must have arrived, independently} \\
    &= \sum_{\substack{X\subset S\nonumber \\\abs X\le1}}\Pr\s*{E^{S\setminus X}_t\cap F^X_t}\prod_{u\in X}\beta_{t,u} \nonumber \\
    &\le \sum_{X\subset S}\Pr\s*{E^{S\setminus X}_t\cap F^X_t}\prod_{u\in X}\beta_{t,u} \tag{Equation~\ref{eq:beta-bound}} \nonumber \\
    &= \sum_{Y\subset S}\Pr\s*{E^{S\setminus Y}_t}\p*{\prod_{u\in Y}\beta_{t,u}}\prod_{u\in S\setminus Y}\p*{1-\beta_{t,u}} \tag{Lemma~\ref{lemma:incl-excl}} \nonumber \\
    &\le \sum_{Y\subset S}\p*{\prod_{u\in S\setminus Y}\alpha_{t,u}}\p*{\prod_{u\in Y}\beta_{t,u}}\prod_{u\in S\setminus Y}\p*{1-\beta_{t,u}} \tag{inductive hypothesis, and Equation~\ref{eq:beta-bound}} \nonumber \\
    &= \prod_{u\in S}\p*{\beta_{t,u}+\alpha_{t,u}(1-\beta_{t,u})} \nonumber \\
    &= \prod_{u\in S}\frac{x_{(v_t,u)}+\alpha_{t,u}(1-\alpha_{t,u}-x_{(v_t,u)})}{1-\alpha_{t,u}} \nonumber \\
    &= \prod_{u\in S}(x_{(v_t,u)}+\alpha_{t,u}) \nonumber ,
  \end{align}
  which is exactly Equation~\ref{eq:thm-corr-main-wts}.

  Before we complete our proof, we must still consider $S\subset B$ where, for at least some $u\in S$, the inequality from Equation~\ref{eq:ass-unsat} is violated.
  Let $S'\subset S$ denote the vertices satisfying Equation~\ref{eq:ass-unsat}.
  Then
  \begin{align*}
    \Pr[E^S_{t+1}] \le \Pr[E^{S'}_{t+1}] \le \prod_{u\in S'}\alpha_{t+1,u} \le \prod_{u\in S}\alpha_{t+1,u} ,
  \end{align*}
  since for any $u\in S\setminus S'$, we have $\alpha_{t+1,u}>1$ (and all other $\alpha_{t,u}$ are non-negative).
\end{proof}

\begin{lemma}\label{lemma:no-proposals-bound}
  In Algorithm~\ref{alg}, at any time $t$, for any subset of vertices $S\subset B$, the probability that none of the vertices in $S$ joins $P$ is upper-bounded by $$\prod_{u\in S} (1-x_{(t,u)}/p_t).$$
\end{lemma}

\begin{proof}
  Pick some $S\subset B$.
  In order for none of the vertices from $S$ to join $P$, then each $u\in S$ that is still unmatched must fail to be sampled with probability $\frac{x_{(v_t,u)}}{p_t(1-α_u)}=β_{t,u}/p_t$, using $β_{t,u}$ as defined in the proof of Lemma~\ref{lemma:corr}.
  We can thus bound this probability to be at most
  \begin{align*}
    &\phantom{{}={}} \sum_{X\subset S}\Pr\s*{E^{S\sm X}_t\cap F^X_t}\prod_{u\in X}(1-β_{t,u}/p_t) \\
    &= \sum_{Y\subset S}\Pr\s*{E^{S\sm Y}_t}\p*{\prod_{u\in Y}(1-β_{t,u}/p_t)}\prod_{u\in S\sm Y}β_{t,u}/p_t \tag{Lemma~\ref{lemma:incl-excl}} \\
    &\le \sum_{Y\subset S}\p*{\prod_{u\in S\sm Y}α_{t,u}}\p*{\prod_{u\in Y}(1-β_{t,u}/p_t)}\prod_{u\in S\sm Y}β_{t,u}/p_t \tag{Lemma~\ref{lemma:corr}} \\
    &= \prod_{u\in S}(α_{t,u}β_{t,u}/p_t+1-β_{t,u}/p_t) \\
    &= \prod_{u\in S}\frac{α_{t,u}x_{(t,u)}/p_t+(1-α_{t,u})-x_{(t,u)}/p_t}{1-α_{t,u}} \\
    &= \prod_{u\in S}(1-x_{(t,u)}/p_t)
      . \qedhere \end{align*}
  \end{proof}

  \begin{lemma}\label{lemma:at-least-w-bound}
    For any  vertex $v_t\in A$,  and any non-negative number $w$ we have:
    \begin{align}
      \Pr[\alg(v_t) \geq w]  \geq  (1- 1/e)\, \sum_{\mathclap{e\ni v_t, w_e\geq w}} x_e.
    \end{align}
  \end{lemma}

  \begin{proof}
    Fix some $v_t\in A$ and $w\ge0$.
    Let $S\subset B$ denote the set of all $u\in B$ such that $w_{(v_t,u)}\ge w$.
    Then as long as some vertex from $S$ is added to $P$ at time $t$, and vertex $v_t$ arrives, then $\alg(v_t)\ge w$ will hold.
    Note that the arrival of $v_t$ is independent of $P$, so we can compute
    \begin{align*}
      \Pr\s*{\alg(v_t)\ge w}
      &= p_t\Pr\s*{S\cap P\ne\emptyset} \\
      &= p_t-p_t\prod_{u\in S}(1-x_{(t,u)}/p_t) \tag{Lemma~\ref{lemma:no-proposals-bound}} \\
      &\ge p_t-p_t\p*{1-\sum_{u\in S}\frac{x_{(t,u)}}{p_t\abs S}}^{\abs S} \tag{AM-GM} \\
      &\ge p_t-p_t\exp\s*{-\sum_{u\in S}x_{(t,u)}/p_t} \\
      &\ge p_t-p_t\p*{1-(1-1/e)\sum_{u\in S}x_{(t,u)}/p_t} \tag{convexity} \\
      &= (1-1/e)\sum_{u\in S}x_{(t,u)}
        . \qedhere
    \end{align*}
  \end{proof}

  \begin{theorem}\label{thm:main1}
    Algorithm~\ref{alg} outputs a $(1-1/e)$-approximate matching, that is
    $$\sum_{v_t\in A} \E[\alg(v_t)] \geq (1-1/e)\cdot \opt.$$
  \end{theorem}

  \begin{proof}
    By Observation~\ref{obs:optlp} we know that $\lpopt$ gives us an upper-bound for $\opt$, that is:
    $$\sum_{e\in E} w_ex_e = \sum_{v_t\in A} \sum_{e\ni v_t} w_ex_e \geq \opt.$$
    As a result, to prove this theorem, it suffices to show  that for any vertex $v_t\in A$ we have
    $$\E[\alg(v_t)]\geq (1-1/e)\sum_{e\ni v_t} w_ex_e,$$ which is the same as proving
    \begin{align}\label{eq:joifwhuwef}
      \sum_{e\ni v_t} w_ex_e-\E[\alg(v_t)] \leq \left(\sum_{e\ni v_t} w_ex_e\right)/e.\end{align}
    By definition, for any vertex $v_t\in A$ we can write the left-hand side of this inequality as

    \begin{align*}
      \sum_{e\ni v_t} w_ex_e-\E[\alg(v_t)] &=  \int_{w=0}^{\infty} \left ( \sum_{e\ni v_t, w_e>w} x_e\right )-  \int_{w=0}^{\infty} \Pr[\alg(v_t) > w] \\
                                        &= \int_{w=0}^{\infty} \left ( \sum_{e\ni v_t, w_e>w} x_e - \Pr[\alg(v_t) > w]\right)\\
                                        & \leq  \int_{w=0}^{\infty} \left ( \sum_{e\ni v_t, w_e>w} x_e\right)/e \tag{Lemma~\ref{lemma:at-least-w-bound} }\\
                                        & = \left(\sum_{e\ni v_t} w_ex_e\right)/e.
    \end{align*}
    This proves Equation~\ref{eq:joifwhuwef} and concludes the proof of the theorem.
  \end{proof}

\subsection{Positive Correlation}\label{section:correlation}

Much of the detail needed in our proof of Theorem~\ref{thm:main1} and related lemmas is due to handling potential correlation between the matched/unmatched status of the vertices in $B$.
In particular, the proof of our main lemma (Lemma~\ref{lemma:no-proposals-bound}) could proceed fairly directly if we were allowed to assume that events $E^{\cu {u_1}}_t$, \dots, $E^{\cu {u_{|B|}}}_{T}$ (the events that vertices in $u\in B$ are matched before any time $t$), are independent from each other. Similarly, if we had the notion of negative dependence used in~\cite{DBLP:conf/sigecom/PapadimitriouPS21}, namely \emph{negative association} of the indicator variables for the events $E^{\cu u}_t$, this would also suffice to arrive at Lemma~\ref{lemma:no-proposals-bound}.
In this section, we will show that a more involved analysis such as ours is in fact necessary since our algorithm sometimes causes positive correlation between these events.
We construct a bipartite graph $G=(A,B,E)$ with $A=\cu{v_1,v_2,v_3}$ and $B=\cu{u_1,u_2}$, such that before time $t=3$  the events $E^{\cu{u_1}}_3$ and $E^{\cu{u_2}}_3$ are in fact positively correlated.
The edge set, along with the values of $p_v$ and $x_e$ for $v\in A$ and $e\in E$ are given in the following diagram:
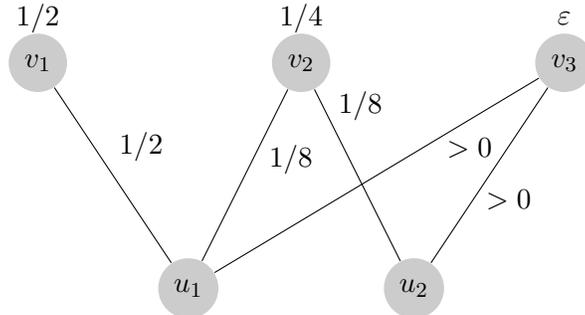
\[\begin{tikzpicture}
    [scale=1,auto=left,every node/.style={circle,fill=black!20}]
    \node [fill opacity=0,text opacity=1] (v1p) at (-5,1.6) {$1/2$};
    \node (v1) at (-5,1) {$v_1$};
    \node [fill opacity=0,text opacity=1] (v2p) at (-1.5,1.6) {$1/4$};
    \node (v2) at (-1.5,1) {$v_2$};
    \node [fill opacity=0,text opacity=1] (v3p) at (2,1.6) {$ε$};
    \node (v3) at (2,1) {$v_3$};

    \node (u1) at (-3,-2) {$u_1$};
    \node (u2) at (0,-2) {$u_2$};

    \draw (v1) -- (u1) node [midway, fill opacity=0,text opacity=1] {$1/2$};
    \draw (v2) -- (u1) node [near start, fill opacity=0,text opacity=1] {$\!1/8$};
    \draw (v3) -- (u1) node [near start, fill opacity=0,text opacity=1] {$\!\!\!\!>0$};
    \draw (v2) -- (u2) node [near start, fill opacity=0,text opacity=1] {$\!1/8$};
    \draw (v3) -- (u2) node [midway, fill opacity=0,text opacity=1] {$\!\!>0$};
  \end{tikzpicture}\]
One can easily verify that our solution $x$ satisfies the LP constraints, and is optimal for certain values of $w$ (in particular, when $w_{(v_1,u_1)}=100$, and when $w_{(v_2,u_1)}=2$ and $w_{(v_2,u_2)}=1$, and small weights incident to $v_3$).

With probability $1/2$, the first vertex $v_1$ arrives, and is matched with $u_1$ with probability $\frac{1/2}{1/2}=1$.
Assuming this occurs, $v_2$ matches $u_2$ if it arrives (with probability $1/4$) and $u_2$ is added to $P$ in the second step (with probability $\frac{1/8}{1/4}=1/2$).

The other $1/2$ of the time, the first vertex $v_1$ does not arrive, so $u_1$ is added to $P$ in the second step with probability $\frac{1/8}{(1/4)(1-1/2)}=1$ when $v_2$ arrives with probability $1/4$, and no edges are matched otherwise.
Overall, after time $t=2$, both $u_1$ and $u_2$ are matched with probability $1/8$, neither is matched with probability $3/8$, and just $u_1$ is matched with probability $3/8+1/8=1/2$.
The indicator variables for the events $E^{\cu{u_1}}_3$ and $E^{\cu{u_2}}_3$ will thus have positive covariance $$\Pr\s*{E^{\cu{u_1,u_2}}_3}-\Pr\s*{E^{\cu{u_1}}_3}\Pr\s*{E^{\cu{u_2}}_3}=1/8-(5/8)(1/8)=3/64>0.$$

\subsection{Tightness of the Analysis}\label{section:tight}
First, we show that our algorithm indeed loses the factor of $(1-1/e)$ compared to $\opt$.
We construct the graph $G=(A,B,E)$, where $A=\cu{v_1,\dc,v_n,v_*}$ and $B=\cu{u_1,\dc,u_n}$.
For each $i$, there are edges $(v_i,u_i)$ and $(v_*,u_i)$ with weights $1/n^2$ and $1$ respectively.
The vertices from $A$ arrive in order $v_1,\dc,v_n,v_*$, and we have $p_{v_i}=1-1/n$ for all $i$ and $p_{v_*}=1$.
Then, the unique optimal solution $x$ to our LP is given in the following diagram (namely, $x_{(v_i,u_i)}=1-1/n$ and $x_{(v_*,u_i)}=1/n$).
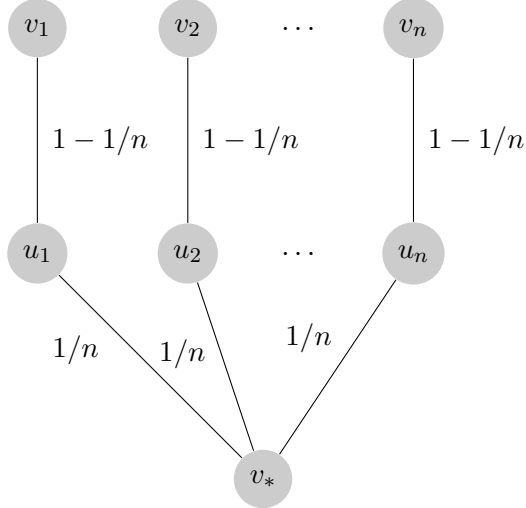
\[\begin{tikzpicture}
    [scale=1,auto=left,every node/.style={circle,fill=black!20}]
    \node (v1) at (-5,1) {$v_1$};
    \node (v2) at (-3,1) {$v_2$};
    \node (vn) at (0,1) {$v_n$};

    \node (vv) at (-2,-5) {$v_*$};

    \node (u1) at (-5,-2) {$u_1$};
    \node (u2) at (-3,-2) {$u_2$};
    \node (un) at (0,-2) {$u_n$};

    \node [fill opacity=0,text opacity=1] at ($(v2)!.5!(vn)$) {\ldots};
    \node [fill opacity=0,text opacity=1] at ($(u2)!.5!(un)$) {\ldots};

    \foreach \from/\to/\wt in {v1/u1,v2/u2,vn/un}
    \draw (\from) -- (\to) node [midway, fill opacity=0,text opacity=1] {$1-1/n$};

    \foreach \from/\to/\wt in {vv/u1,vv/u2}
    \draw (\from) -- (\to) node [near end, fill opacity=0,text opacity=1] {$1/n$};
    \foreach \from/\to/\wt in {vv/un}
    \draw (\from) -- (\to) node [midway, fill opacity=0,text opacity=1] {$1/n$};
  \end{tikzpicture}\]
Consider what our algorithm would do faced with this graph. For each $i$, it would add $u_i$ to $P$ with probability $1$, then add $(v_i,u_i)$ to our matching if $v_i$ is realized.
Hence, the probability that $u_i$ is unmatched by the time we get to vertex $v_*$ is exactly $1/n$, and is independent of all other vertices from $B$.
The probability that $v_*$ will have no neighbors unmatched is thus $(1-1/n)^n$.

We can now bound our algorithms expected matching weight to be at most $$\p*{1-(1-1/n)^n}+n(1-1/n)(1/n^2),$$ which for large $n$ gets arbitrarily close to $1-1/e$.
On the other hand, a trivial online algorithm could instead never match any of the edges $(v_i,u_i)$, and always take one of the edges $(v_*,u_i)$, obtaining a matching with weight $1$ always.
The optimal online algorithm, $\opt$, would thus need to attain at least $1$ in expectation, proving that our algorithm can never be $(1-1/e+ε)$-competitive for any $ε>0$. 

Note that in this example, for any vertex $A$ at the time of its arrival, the matching status of its neighbors are independent. This intuitively means that the loss our algorithm incurs is not due to the correlation it causes between the matching status of the vertices.

\subsection{Generalization of the Algorithm and Analysis}\label{section:general}
With our analysis complete, we can now extend our algorithm to a more general version of the vertex arrival model, allowing for distributions over edge weights.
We will first describe the new model, which Papadimitriou, Pollner, Saberi, and Wajc also considered for their algorithm~\cite{DBLP:conf/sigecom/PapadimitriouPS21}.
We will then explain the main considerations for adapting our algorithm and analysis from Sections~\ref{section:alg}--\ref{sec:analysis} for this harder case.

\paragraph{General Vertex Arrival Model}
Just as in the original vertex arrival model (described in Section~\ref{section:pre}), we have a known bipartite graph $G=(A,B,E)$ and fixed order $(v_1,\dotsc,v_{\abs A})$ over the vertices in $A$.
Vertices in $B$ are initially present, but vertices in $A$ arrive online in this order.
However, rather than having fixed weights for all edges, with each vertex $v_t\in A$ arriving with a probability $p_{v_t}$, we instead realize a sample $w^t$ from a distribution over possible weights for all edges incident to $v_t$.

Formally, for each time $1\le t\le\abs A$, there is a joint distribution $\mc D_t$ with finite support over non-negative assignments of weights for all edges incident to $v_t$.
At time $t$, we draw a sample $w^t\sim\mc D_t$.
This tells us the realized weight $w^t_e$ for each edge $e=(v_t,u)$ incident to $v_t$.
As before, we may now choose to match $v_t$ irrevocably to one of its unmatched neighbours.
The goal is to maximize the total realized weight of all edges in our matching, given by
\[\sum_{(v_t,u)\in M}w^t_{(v_t,u)} ,\]
where $M$ denotes the set of edges in our final matching.

We note that this is indeed a generalization of our original vertex arrival model, which can be represented here by $\mc D_t$ yielding the vector of values $(w_{(v_t,u)})_u$ with probability $p_t$, and the zero vector with probability $1-p_t$.

\paragraph{Modified Algorithm}
To begin, we modify our LP from Section~\ref{section:alg} to yield another LP relaxation under this more general model, using the same natural extension as given in~\cite{DBLP:conf/sigecom/PapadimitriouPS21}.

Since our distributions $\mc D_t$ are assumed to be finite, for each $t$, we let $p_{t,i}$ denote the probability mass for the $i$-th possibility (and $i$ varies from $1$ to the size of the support of $\mc D_t$), and define $w_{t,i,u}$ to be the weight assigned to edge $(v_t,u)$ in this case.
We will have variables $y_{t,i,u}$ representing the probability of $v_t$ being matched to $u$ with value $w_{t,i,u}$ in $\opt$.
These take the place of $x_e$ from before (representing the probability of $e$ being in our matching), so we can now give our modified LP:
\begin{align}
\max_{\mathbf{y}} \qquad &\sum_{(v_t,u)\in E, i} w_{t,i,u}\cdot y_{t,i,u}	\,,& \label{eq:LPg-objfun} \\
\textrm{s.t.} \qquad & \sum_{(v_t,u)\in E} y_{t,i,u} \le p_{t,i} &  \forall\, v_t\in A, i  \,, \label{eq:LPg-topbound} \\
& \sum_{t,i} y_{t,i,u} \le 1 &  \forall\, u\in B  \,, \label{eq:LPg-botbound} \\
& p_{t,i}\cdot (1-\sum_{t'<t,i'} y_{t',i',u}) \ge y_{t,i,u} &  \forall\, v_t\in A,u\in B,i  \,, \label{eq:LPg-main} \\
&y_{t,i,u} \ge 0& \forall v_t\in A,u\in B,i \,. \label{eq:LPg-nonneg}
\end{align}

The actual rounding procedure of Algorithm~\ref{alg} also needs modification.
At each iteration of the loop, we will first sample $w^t\sim\mc D_t$, obtaining some possibility $\hat i$.
We can define
\begin{align}
  \alpha_u=\sum_{t'<t,i}y_{t',i,u} \label{eq:genalpha}
\end{align}
instead at Line~\ref{alg:defalpha}.
We will now add each vertex $u\in N_v$ to $P$ independently with probability $\frac{y_{t,\hat i,u}}{p_{t,\hat i}(1-\alpha_u)}$ instead of $\frac{x_{(v,u)}}{p_v(1-\alpha_u)}$ at Line~\ref{alg:samp-line}.
This is again easily seen to be well-defined, by Equation~\ref{eq:LPg-main} from our modified LP.
The maximum weight sampled neighbour will then be chosen based on the sampled weights $w_{t,\hat i,u}$.

\paragraph{Modified Analysis}
Our analysis remains largely valid, and applies to this more general model with minor modifications.

To start, $\alpha_{t,u}$ must be defined as in Equation~\ref{eq:genalpha}.
This allows the proof of Lemma~\ref{lemma:corr} to go through as written, replacing occurrences of $x_{(v_t,u)}$ with $\sum_iy_{t,i,u}$.
We must also more carefully expand the probability at Equation~\ref{eq:corr-line-expand}, observing that
\begin{align*}
  \Pr\s*{(v_t,u)\in M\mid E^{S\sm\cu u}_t\cap F^{\cu u}_t}
  \le \sum_ip_{t,i}\frac{y_{t,i,u}}{p_{t,i}(1-\alpha_{t,u})}
  = \beta_{t,u}.
\end{align*}

Lemma~\ref{lemma:no-proposals-bound} needs slight adjustment to its statement.
We can instead show that, for any $i$, \emph{if we assume} that the $i$-th possible weight vector is drawn from $\mc D^t$, so the realized weight of $(v_t,u)$ is $w_{t,i,u}$ for all $u$, \emph{then} the probability that no vertex from $S$ joins $P$ is upper-bounded by
\[\prod_{u\in S}(1-y_{t,i,u}/p_{t,i}).\]
The proof now still holds, replacing all occurrences of $x_{(t,u)}$ with $y_{t,i,u}$, and replacing occurrences of $p_t$ with $p_{t,i}$.

Lemma~\ref{lemma:at-least-w-bound} can be modified similarly, again conditioning on the realization of $\mc D^t$, and making the same substitutions, defining $S\subset B$ as all $u\in B$ where $w_{t,i,u}\ge w$ for the assumed realization $i$.
Finally, to extend Theorem~\ref{thm:main1} to this more general model, by taking expectations over the drawing of $w^t\sim\mc D^t$, it suffices to show for any time $t$ and realization $i$ that
\[\E[\alg(v_t)\mid w^t_u=w_{t,i,u}]\geq (1-1/e)\sum_{(v_t,u)\in E} w_{t,i,u}y_{t,i,u}.\]
Again, the proof carries out similarly to before, using our conditional version of Lemma~\ref{lemma:at-least-w-bound}, and replacing occurrences of $w_{(v_t,u)}$ and $x_{(v_t,u)}$ with $w_{t,i,u}$ and $y_{t,i,u}$, respectively.

\section{Edge Arrivals}\label{sec:edge}

\subsection{The Algorithm}\label{sec:edge-algorithm} 

Similar to the vertex arrival version, we start with  an LP for the online problem, and use its solution to build our matching. Again, for each edge $e\in E$, we have the variable $x_e$ represent the probability of $e$ joining the matching in $\opt$.
\begin{align} 
\max_{\mathbf{x}} \qquad &\sum_{e\in E} w_ex_e	\,,& \label{eq:LP-edge-objfun} \\
\textrm{s.t.} \qquad & \sum_{e\ni u} x_e \le 1 &  \forall\, u\in A\cup B  \,, \label{eq:LP-edge-sumone} \\
& p_e\cdot (1-\sum_{\mathclap{e'\ni v, e'<e}} x_{e'}) \ge x_{e} &  \forall\, e=(v,u)\in E  \,, \label{eq:LP-edge-main-top} \\
& p_e\cdot (1-\sum_{\mathclap{e'\ni u, e'<e}} x_{e'}) \ge x_{e} &  \forall\, e=(v,u)\in E  \,, \label{eq:LP-edge-main-bot} \\
&x_e \ge 0& \forall e \in E \,. \label{eq:LP-edge-nonneg}
\end{align}
We would again like to assert that any $x_e$ corresponding to the execution of $\opt$ yields a valid solution to this LP.
Constraint~\ref{eq:LP-edge-sumone} is as before.

We now consider Constraint~\ref{eq:LP-edge-main-top}.
In order for $\opt$ to add $e=(v,u)$ to the matching, it cannot have matched any edge to $v$ already.
This occurs with probability exactly $\sum_{e'\ni v,e'<v}x_{e'}$, by definition of $e'$, and since the corresponding events are disjoint.
Finally, since $p_e$ being realized is independent from all previous realizations (and any randomness used by the algorithm), the probability that $v$ has not been matched \emph{and} $e$ is realized is given by the left-hand side of Constraint~\ref{eq:LP-edge-main-top}, and so the bound must follow.
Constraint~\ref{eq:LP-edge-main-bot} is similar, and we obtain an observation analogous to Observation~\ref{obs:optlp}.
\begin{observation}\label{obs:optlp-edge}
  Let \b{x} be an optimal solution of the LP. We have $\lpopt \geq \opt$ where $\lpopt = \sum_{e\in E} w_ex_e$.
\end{observation}

\subsection{The Rounding Procedure}
\label{sec:edge-rounding}

We give our online rounding procedure in Algorithm~\ref{edge-alg}.
Here, we think of the vertices $u\in B$ as again making proposals to their neighbours $v\in A$ with some probability (based on $b$), as long as the corresponding edge $e_i$ is realized.
Then, $v$ must decide if it accepts a proposal online. This is as opposed to the vertex arrival model, where $v$ knew all its proposals upon arrival. Since the graph is weighted, simply accepting the first proposal may result in a significant loss. To resolve this issue, our algorithm is designed in a way that each edge $e$ joins the final matching with probability exactly $x_e/2$. In this sense, our algorithm  resembles the one designed by Ezra et al.~\cite{DBLP:conf/sigecom/EzraFGT20} for the vertex arrival version of the problem. 

Before stating the algorithm formally, we give a brief overview. The algorithm starts with all the vertices marked as alive, but as the algorithm proceeds it marks some of them as dead. Vertices in $B$ only die when they are matched. However, we sometimes mark a vertex in $A$ as dead without it being matched. At any time $t$, when edge $e=(v, u)$ arrives, the algorithm needs to decide whether to add this edge to the matching. If $v$ is alive at this point, independent of the status of $u$, it randomly (with a probability set in  Line~\ref{line:ihurfihuerf} of the algorithm) decides whether to send a proposal to $u$. The probability of this event is set in a way that the probability of $u$ ever receiving a proposal from $v$ is equal to $x_e$. If $u$ is alive, it decides randomly (with a probability set in  Line~\ref{line:jhrhgjre} of the algorithm) whether to accept the proposal. If a match happens, we mark $u$ as dead to ensure that we do not match it again in the future. However, vertex $v$ dies iff it send a proposal regardless of the proposal being accepted. This serves two purposes. First, to ensure that its future edges are not matched with a probability higher than $1/2$. Second, to ensure that alive/dead status of the vertices in $A$ are independent of each other throughout the algorithm. 
\vspace{2mm}

\begin{tboxalg}{Rounding Procedure}\label{edge-alg}
\begin{algorithmic}[1]
\STATE Let $\b x$ be an optimal solution of the LP.
\STATE Let $M\gets \emptyset$ be a matching of $E$.
\STATE Mark all the vertices in $V$ as alive.
\FOR{$t\in |E|$}
\STATE Let $e_t = (v,u)$ where $v\in A$ and $u\in B$.
\STATE Define $\alpha_{u}=\sum_{e\ni u, e<e_t} x_{e}.$ \label{line:edgealg-alphau}
\STATE Define $\alpha_{v}=\sum_{e\ni v, e<e_t} x_{e}.$
\STATE Let $b$ be a Bernoulli random variable which is equal to one with probability  $\frac{x_{e_t}}{p_{e_t}(1-\alpha_u)}$.\label{line:ihurfihuerf}
\IF{$u$ is alive, $e_t$ is realized, and $b=1$} \label{line:edgealg-cond}
\STATE If $v$ is also alive, then with probability $\frac{1}{2-\alpha_v}$ add $e_t$ to $M$ and mark $v$ as dead.\label{line:jhrhgjre}
\STATE Mark $u$ as dead.
\ENDIF
\ENDFOR
\STATE Return matching $M$.
\end{algorithmic}
\end{tboxalg}
\vspace{2mm}
We note that this algorithm necessarily returns a valid matching since whenever we add an edge $e_i=(v,u)$ to $M$, we also mark both $v$ and $u$ as dead (and will never again add any of their incident edges to $M$).
Otherwise, everything is well-defined (notably, $\alpha_v,\alpha_u\in[0,1]$) by the LP constraints.

\subsection{The Analysis} \label{section:edge-analysis}

The first half of our analysis will focus on showing that the proposals arriving at a given vertex $u$ are well-behaved.
To begin, we show that a vertex $u\in B$ proposes to $v$ with probability exactly $x_{(v,u)}$.
\begin{lemma}\label{lem:edge-proposal-prob}
  On any iteration $t$ of Algorithm~\ref{edge-alg}, the probability that the condition at Line~\ref{line:edgealg-cond} holds is $x_{e_t}$.
\end{lemma}
\begin{proof}
  We prove this by strong induction for a given vertex $u\in A$.
  Fix $t\ge1$, and suppose this holds for all $t'<t$.
  That is, for every $(v,u_{t'})<t$, the probability that the condition at Line~\ref{line:edgealg-cond} holds (that is, the probability that $u$ proposes to $v_{t'}$) is $x_{e_{t'}}$.
  Then, defining $α_{t,u}:=\sum_{u\ni e,e<(v_t,u)}x_e$ as computed at Line~\ref{line:edgealg-alphau}, the probability that $u$ is dead at the start of iteration $t$ is exactly $α_{t,u}$, since $u$ is marked dead as soon as it makes its first (and thus only) proposal.

  Whether $e_t$ is realized and whether $b=1$ at iteration $t$ both occur independently of what has occurred so far, and with probabilities $\frac{x_{e_t}}{p_{e_t}(1-α_{t,u})}$ and $p_{e_t}$ respectively.
  Thus, the probability that all three conditions from Line~\ref{line:edgealg-cond}, and that $u$ proposes to $v_t$, is exactly $x_{e_t}$.
\end{proof}
Next, we observe that for a fixed $v\in A$, the proposals received from its neighbors $u\in B$ are independent.
\begin{lemma}\label{lem:edge-proposal-independence}
  For an edge $e_t\in E$, let $P_{e_t}$ denote the event that in iteration $t$, the condition at Line~\ref{line:edgealg-cond} holds.
  Then for any $v\in A$, the events $\cu{P_{(v,u)}:(v,u)\in E}$ are independent.
\end{lemma}
\begin{proof}
  Let $e_t=(v,u)$.
  We observe that $P_{(v,u)}$ depends only on randomness from iteration $t$ (whether $e_t$ was realized, the value of $b$), as well as whether $u$ is alive or not.
  However, the aliveness of $u$ itself depends on these same variables from the previous edge incident to $u$ processed by the algorithm (or $u$ is deterministically alive if $e_t$ is the first such edge).
  Thus, inducting over all such edges, whether $u$ is alive or not depends only on realizations of edges and values of $b$ from iterations $t'$ where $e_{t'}=(v',u)$ for some $v'$.
  Importantly, $P_{(v,u)}$ is a deterministic function of these random inputs, which are importantly disjoint from $P_{(v,u')}$ for other $u'\ne u$.
\end{proof}

Now that we know that the proposals are well-behaved, we can prove our main result for edge arrivals.

\begin{theorem}\label{thm:main2}
  Algorithm~\ref{edge-alg} outputs a $1/2$-approximate matching, that is
  $$\sum_{e\in E} w_e\E[\alg(e)] \geq \opt/2.$$
\end{theorem}
\begin{proof}
  We have already noted in Section~\ref{sec:edge-rounding} that Algorithm~\ref{edge-alg} outputs a correct matching.
  It thus suffices to prove that each edge $e=(v,u)$ is added to $M$ with probability $x_e/2$, by Observation~\ref{obs:optlp-edge}.
  For a given $v$, we prove this by induction over all edges incident to $v$, in order of arrival.

  Fix some $t\ge1$ where $e_t=(v,u)$.
  Suppose that any $e_{t'}=(v,u')$ with $t'<t$ is added with probability $x_{e_{t'}}/2$.
  Then, at time $t$, the probability that $v$ is already marked dead (equivalently, that an edge incident to $v$ has been added to $M$) is exactly $\sum_{e\ni v,e<t}x_e/2=α_v/2$, as these are disjoint events.
  By Lemma~\ref{lem:edge-proposal-independence}, the proposals to $v$ were independent, so even conditioned on $v$ being still alive, the probability that $v$ receives a proposal from $u$ is as given by Lemma~\ref{lem:edge-proposal-prob}, namely $x_e$.
  Thus, the probability that $v$ is alive and proposed to by $u$, and $(v,u)$ is then added to $M$, is
  \[(1-α_v/2)\cdot x_e\cdot\frac1{2-α_v}=x_e/2 . \qedhere \]
\end{proof}

\section{Proof of Lemma~\ref{lemma:incl-excl}}\label{apx:omitted}
\begin{proof}
  We begin by considering the events $E$ and $F$.
  Throughout, we assume a single fixed $t$, and drop it from the subscripts.
  First, for any $X\subset B$, the set of events $E^{X\sm Y}\cap F^Y$ over $Y\subset X$ partition the probability space.
  In particular, we get the identity
  \begin{align*}
    \sum_{Y\subset X}\Pr\s*{E^{X\sm Y}\cap F^Y}
    &= 1
      .
  \end{align*}
  Even better, this same identity holds when we condition all probabilities by an arbitrary event, so for $X\subset S\subset B$ we have
  \begin{align}
    \label{eq:incl-excl-sum2}
    \sum_{Y\subset X}\Pr\s*{E^{S\sm Y}\cap F^Y}
    = \sum_{Y\subset X}\Pr\s*{E^{S\sm X}}\Pr\s*{E^{X\sm Y}\cap F^Y\mid E^{S\sm X}}
    &= \Pr\s*{E^{S\sm X}}
  \end{align}
  by conditioning on $E^{S\sm X}$.

  Letting $S\subset B$ and $w\in\mathbb R^S$ be arbitrary, we now get
  \begin{align*}
    &\phantom{{}={}} \sum_{X\subset S}\Pr\s*{E^{S\sm X}}\p*{\prod_{u\in X}w_u}\prod_{u\in S\sm X}(1-w_u) \\
    &= \sum_{X\subset S}\p*{\sum_{Y\subset X}\Pr\s*{E^{S\sm Y}\cap F^Y}}\p*{\prod_{u\in X}w_u}\prod_{u\in S\sm X}(1-w_u) \tag{Equation~\ref{eq:incl-excl-sum2}} \\
    &= \sum_{Y\subset S}\Pr\s*{E^{S\sm Y}\cap F^Y}\sum_{Y\subset X\subset S}\p*{\prod_{u\in X}w_u}\prod_{u\in S\sm X}(1-w_u) \\
    &= \sum_{Y\subset S}\Pr\s*{E^{S\sm Y}\cap F^Y}\p*{\prod_{u\in Y}w_u}\sum_{Y\subset X\subset S}\p*{\prod_{u\in X\sm Y}w_u}\prod_{u\in S\sm X}(1-w_u) \\
    &= \sum_{Y\subset S}\Pr\s*{E^{S\sm Y}\cap F^Y}\p*{\prod_{u\in Y}w_u}\prod_{u\in S\sm Y}(w_u+(1-w_u)) \\
    &= \sum_{Y\subset S}\Pr\s*{E^{S\sm Y}\cap F^Y}\p*{\prod_{u\in Y}w_u}
      . \qedhere \end{align*}
\end{proof}

\newpage
\bibliographystyle{plain}
\bibliography{refs}

\appendix

\end{document}